\documentclass[aps,twocolumn,amsmath,amssymb,nofootinbib,superscriptaddress]{revtex4-1}
\usepackage{amsmath,amsthm,amsfonts,mathptmx,fullpage,diagbox,subfigure}
\usepackage{tablefootnote}
\usepackage[margin=0.7in]{geometry}
\usepackage[hang,small,bf]{caption}    
\usepackage{qcircuit}
\usepackage{tikz}	
\usepackage{tablefootnote} 
\usetikzlibrary{backgrounds,fit,decorations.pathreplacing}  

\def\c1{{\ctrl{1} }}
\def\c2{{\ctrl{2} }}
\def\b1{{\ctrl{-1} }}
\def\b2{{\ctrl{-2} }}

\usepackage{hyperref} 

\usepackage{graphicx,color,cases} 
\usepackage{xfrac}  
\definecolor{DarkGreen}{rgb}{0.1,0.5,0.1}
\definecolor{DarkRed}{rgb}{0.5,0.1,0.1}
\definecolor{DarkBlue}{rgb}{0.1,0.1,0.5} 
 
\newtheorem{theorem}{Theorem} 
\newtheorem{lemma}[theorem]{Lemma}

\theoremstyle{definition}

\numberwithin{equation}{section}
 
\def\>{\rangle} 
\def\<{\langle}

\newcommand{\qbmat}[4]{\begin{pmatrix}  {#1} & {#2} \\ {#3} & {#4} \\\end{pmatrix}}

\newcommand{\nc}{\newcommand}

\nc{\bbA}{\mathbb{A}} \nc{\bbB}{\mathbb{B}} \nc{\bbC}{\mathbb{C}}
\nc{\bbD}{\mathbb{D}} \nc{\bbE}{\mathbb{E}} \nc{\bbF}{\mathbb{F}}
\nc{\bbG}{\mathbb{G}} \nc{\bbH}{\mathbb{H}} \nc{\bbI}{\mathbb{I}}
\nc{\bbJ}{\mathbb{J}} \nc{\bbK}{\mathbb{K}} \nc{\bbL}{\mathbb{L}}
\nc{\bbM}{\mathbb{M}} \nc{\bbN}{\mathbb{N}} \nc{\bbO}{\mathbb{O}}
\nc{\bbP}{\mathbb{P}} \nc{\bbQ}{\mathbb{Q}} \nc{\bbR}{\mathbb{R}}
\nc{\bbS}{\mathbb{S}} \nc{\bbT}{\mathbb{T}} \nc{\bbU}{\mathbb{U}}
\nc{\bbV}{\mathbb{V}} \nc{\bbW}{\mathbb{W}} \nc{\bbX}{\mathbb{X}}
\nc{\bbZ}{\mathbb{Z}}

\nc{\bA}{{\bf A}} \nc{\bB}{{\bf B}} \nc{\bC}{{\bf C}}
\nc{\bD}{{\bf D}} \nc{\bE}{{\bf E}} \nc{\bF}{{\bf F}}
\nc{\bG}{{\bf G}} \nc{\bH}{{\bf H}} \nc{\bI}{{\bf I}}
\nc{\bJ}{{\bf J}} \nc{\bK}{{\bf K}} \nc{\bL}{{\bf L}}
\nc{\bM}{{\bf M}} \nc{\bN}{{\bf N}} \nc{\bO}{{\bf O}}
\nc{\bP}{{\bf P}} \nc{\bQ}{{\bf Q}} \nc{\bR}{{\bf R}}
\nc{\bS}{{\bf S}} \nc{\bT}{{\bf T}} \nc{\bU}{{\bf U}}
\nc{\bV}{{\bf V}} \nc{\bW}{{\bf W}} \nc{\bX}{{\bf X}}
\nc{\bY}{{\bf Y}} \nc{\bZ}{{\bf Z}}
\nc{\bHP}{{\bf HP}}

\nc{\bmA}{{\bm A}} \nc{\bmB}{{\bm B}} \nc{\bmC}{{\bm C}}
\nc{\bmD}{{\bm D}} \nc{\bmE}{{\bm E}} \nc{\bmF}{{\bm F}}
\nc{\bmG}{{\bm G}} \nc{\bmH}{{\bm H}} \nc{\bmI}{{\bm I}}
\nc{\bmJ}{{\bm J}} \nc{\bmK}{{\bm K}} \nc{\bmL}{{\bm L}}
\nc{\bmM}{{\bm M}} \nc{\bmN}{{\bm N}} \nc{\bmO}{{\bm O}}
\nc{\bmP}{{\bm P}} \nc{\bmQ}{{\bm Q}} \nc{\bmR}{{\bm R}}
\nc{\bmS}{{\bm S}} \nc{\bmT}{{\bm T}} \nc{\bmU}{{\bm U}}
\nc{\bmV}{{\bm V}} \nc{\bmW}{{\bm W}} \nc{\bmX}{{\bm X}}
\nc{\bmZ}{{\bm Z}}

\nc{\cA}{{\cal A}} \nc{\cB}{{\cal B}} \nc{\cC}{{\cal C}}
\nc{\cD}{{\cal D}} \nc{\cE}{{\cal E}} \nc{\cF}{{\cal F}}
\nc{\cG}{{\cal G}} \nc{\cH}{{\cal H}} \nc{\cI}{{\cal I}}
\nc{\cJ}{{\cal J}} \nc{\cK}{{\cal K}} \nc{\cL}{{\cal L}}
\nc{\cM}{{\cal M}} \nc{\cN}{{\cal N}} \nc{\cO}{{\cal O}}
\nc{\cP}{{\cal P}} \nc{\cQ}{{\cal Q}} \nc{\cR}{{\cal R}}
\nc{\cS}{{\cal S}} \nc{\cT}{{\cal T}} \nc{\cU}{{\cal U}}
\nc{\cV}{{\cal V}} \nc{\cW}{{\cal W}} \nc{\cX}{{\cal X}}
\nc{\cY}{{\cal Y}} \nc{\cZ}{{\cal Z}}

\begin{document}

\title{Faster quantum computation with permutations and resonant couplings}

\author{Yingkai Ouyang}\email[]
{y.ouyang@sheffield.ac.uk (corresponding author)}
\affiliation{Department of Physics \& Astronomy, University of Sheffield, Sheffield, S3 7RH, United Kingdom}
\affiliation{Singapore University of Technology and Design, 8 Somapah Road, Singapore 487372}
\affiliation{Centre for Quantum Technologies, National University of Singapore, 3 Science Drive 2, Singapore 117543}							 
\author{Yi Shen}\email[]
{yishen@buaa.edu.cn (corresponding author)}
\affiliation{School of Mathematics and Systems Science, Beihang University, Beijing 100191, China}

\author{Lin Chen}\email[]
{linchen@buaa.edu.cn (corresponding author)}
\affiliation{School of Mathematics and Systems Science, Beihang University, Beijing 100191, China}
\affiliation{International Research Institute for Multidisciplinary Science, Beihang University, Beijing 100191, China}

\date{\today}

\frenchspacing

\begin{abstract}
Recently, there has been increasing interest in designing schemes for quantum computations that are robust against errors.
Although considerable research has been devoted to quantum error correction schemes, 
much less attention has been paid to optimizing the speed it takes to perform a quantum computation and developing computation models that act on decoherence-free subspaces.
Speeding up a quantum computation is important, because fewer errors are likely to result.
Encoding quantum information in a decoherence-free subspace is also important, because errors would be inherently suppressed.
In this paper, we consider quantum computation in a decoherence-free subspace and also optimize its speed. 
To achieve this, we perform certain single-qubit quantum computations by simply permuting the underlying qubits.
Together resonant couplings using exchange-interactions, 
we present a new scheme for quantum computation that potentially improves the speed in which a quantum computation can be done. 
\end{abstract}

\maketitle

\section{Introduction}
By harnessing the powers of quantum mechanics, quantum computers can potentially unleash algorithms of unprecedented power, improve the precision of metrology, and enable novel cryptographic schemes.
However the inherent fragility of quantum information frustrates the construction of any quantum computer that is to compute correctly. 
Designing schemes for quantum computation that can proceed with high fidelity is thus an important problem.

Performing a quantum computation within a decoherence-free subspace (DFS) is a natural approach towards combating decoherence, and is applicable in a wide range of physical systems.
Conventionally in an artificial atom, a qubit has its basis states assigned to be a ground state and an excited state respectively.
Because of this, the qubit is vulnerable to picking up an unwanted phase that results entirely from the physical system's natural dynamics.
To avoid this problem, several artificial atoms could encode a logical qubit in a constant energy subspace by assigning the logical basis states to have a constant number of excitations.
Spurious phase errors that would have resulted from the natural dynamics of the system's evolution can thereby be intrinsically avoided.
Such encodings into DFSs have been explored both in the context of quantum computation \cite{KLM01,divincenzo2000universal,Bacon01} and quantum error correction \cite{CLY97,LBW99,BvL16,ouyang2018permutation}.
Combined with carefully tailored active control on the underlying physical system \cite{UhrigPhysRevLett.98.100504,heinze2018universal}, decoherence can be substantially mitigated. 

An observation that one could make is that the longer it takes to perform a quantum computation, the more likely it is for quantum information to decohere.
It is therefore important to speed up a quantum computation to maximize the computation's fidelity. 
This has been discussed in the context of performing quantum computations within DFSs \cite{divincenzo2000universal}, where the speed of a spin-based quantum computation is optimized by sidestepping the need to use slow single qubit gates. In Ref \cite{divincenzo2000universal}, a logical qubit comprises of three physical qubits, and all gates are driven entirely by resonant couplings between pairs of physical qubits. Because these coupling strengths are large, the computation can be fast.
In the work of DiVincenzo {\em et al.}, the speed of the computation can be quantified by the number of timesteps the computation requires.
With respect to this, single-qubit gates and a CNOT gate can be performed in 3 and 11 timesteps respectively.

Inspired by the possibility of performing quantum computations entirely by exchange interactions in Ref \cite{divincenzo2000universal}, Levy devised a scheme for quantum computation using only exchange interactions, but which requires far fewer timesteps \cite{levy2002universal}. 
In Levy's scheme, just two physical qubits form a logical qubit, and single-qubit gates and an entangling two-qubit gate are performed in 3 and 2 timesteps respectively. 
However, Levy's scheme no longer computes within a DFS, and its logical qubits are hence vulnerable to phase noise. 
Similarly, while superconducting qubits \cite{blais2004cavity} can have quantum computation that results from strong resonant couplings and are hence fast, the resultant computation also does not lie within a DFS.
Therein arises the question: can one reduce the number of timesteps required in a quantum computation that proceeds within a DFS?

The subject of how one can perform a quantum computation has been extensively studied.
Of specific interest to us are models of quantum computation based on the quantum circuit model, where the computation is composed from certain single-qubit and two-qubit gates \cite{boykin,KLM01,bravyi2005universal}.

In a markedly different setting, one can indeed implement quantum computations in one of the simplest imaginable ways, which is solely by performing permutations of the underlying particles.
One related model of quantum computation relies on braiding particles \cite{freedman2003topological}; the braiding group for $n$ objects is different and richer than the permutation group of $n$ objects. 
However, quantum computation by braiding uses nonabelian anyons, and remains challenging to implement in a laboratory setting.
In view of this, one might wonder if quantum computation can be performed simply by permuting regular qubits.
 However in most physical implementations, permuting the qubits is not an easy operation. 
What we actually recommend is that physical qubits need not be actually permuted in practice; a classical computer keeps track of all the permutations that take place, essentially relabeling the qubits whenever a permutational operation has to take place. Non-permutational gates then proceed by interacting the corresponding relabeled qubits. This allows permutational gates to be performed by a classical computer, and thereby allow faster quantum computation by this precompiling technique.

The question of permutational quantum computation has been studied by Jordan, who studied the extent in which permutations alone can perform an interesting family of quantum computations on rather specific problems \cite{jordan2009permutational}.
It is however unlikely that permuting qubits can effect an arbitrary quantum computation, because there is only a finite number of permutations while the number of possible quantum computations is infinite.
For example, Jordan gives in \cite[Section 9]{jordan2009permutational} a counting argument explaining why it is unlikely that permutational quantum computing is equal to BQP, the class of decision problems solvable by a quantum computer in polynomial time with a bounded probability of error.
Hence, to achieve a qubit-based universal quantum computation,
permutations must be augmented by non-permutational gates.

Could one perform quantum computations in a DFS using only permutations and resonant couplings?
In this paper, we answer this question in the affirmative. 
We show that a non-trivial set of permutations augmented by realistic resonant couplings can indeed allow for a universal quantum computation to proceed within a DFS.
Such a scheme could reduce the number of timesteps required in the quantum computation. 
This is because a classical computer could keep track of how the underlying qubits are permuted and determine where and when resonant couplings are applied between pairs of physical qubits.
Because of this, parts of a quantum computation can be offloaded to a classical computer, which allows the quantum computation to proceed with fewer timesteps than the scheme of Ref \cite{divincenzo2000universal}.

The caveat of requiring at least 32 physical qubits to encode each logical qubit might not be too severe, because the potential speed ups offered when the quantum computation to be performed is reasonably complicated, such as that in fault-tolerant quantum computation, might be worth the while.

The possibility of performing the `classical part' of a quantum computation by permutations, while intuitive is far from obvious. 
The Clifford group of single qubit gates, generated by a Hadamard and a complex phase gate, is a finite group, 
and hence is by Cayley's theorem a subgroup of a symmetric group of some size \cite{dummit-foote}.
While the explicit structure of the single-qubit Clifford group has been elucidated by Planat \cite{planat2010clifford}, its realization in terms of permuting qubits is far from resolved.
This is because even though the matrix representation of the single-qubit Clifford group in terms of permutations exists, there might not exist a state space on which these permutations act faithfully. 
For example, while the matrix algebra of the single-qubit Clifford gates can be made to be correct, its action is entirely dependent on the basis chosen for qubit's logical encoding.
An arbitrary selection of the basis would invariably lead to an incorrect group action on the chosen basis states. 
We show that certain subgroups of the single-qubit Clifford group can be performed entirely by permutations, and give a lower bound for the number of qubits required to realize single-qubit Clifford computations by permutations.

The scheme that we propose to augment present can be seen to be a discretized version of the dual-rail encoding \cite{KLM01}.
When our scheme encodes 32 physical qubits into a logical qubit, it allows the $\frac \pi 8$-gate and bit flip gates to be performed by permutations.
Exploiting the fact that the bit-flip gate is a product of disjoint swaps, we show that the
Hadamard gate and CNOT can be performed using physically realistic resonant couplings.
To illustrate the potential of our scheme, we calculate the number of timesteps required to implement an important gate in quantum computation, the Toffoli gate.
Our proposed scheme implements the Toffoli gate within a DFS in 82 timesteps, which is slightly faster than the 85 timesteps required in the scheme of DiVincenzo {\em et al.}.

The resonant couplings that we require in our scheme are arguably physically realistic. 
Their purpose is to implement the quantum Fredkin gate and exchange interactions on the physical qubit level. 
Exchange interactions can potentially be implemented accurately even at a non-zero temperature \cite{schuetz2017high} in an experimental setting. The quantum Fredkin gate is a controlled swap gate, and has been extensively studied \cite{chau1995simpleFredkin,patel2016Fredkin,liu2018onestepFredkinSC} since its introduction \cite{milburn1989Fredkin}. While the quantum Fredkin gate is more challenging to implement than simple exchange, recent work suggests that it can be implemented using physically realistic resonant interactions on superconducting qubits in a single timestep \cite{liu2018onestepFredkinSC}.

Our results are theoretical in nature, and do not deal with any specific experimental system. This allows our approach to be potentially applied in a wide range of physical systems once the specifics of these systems are taken into account.
In Section \ref{sec:re}, we present our first scheme which allows universal quantum computation. This scheme is an extended dual-rail scheme which uses 32 qubits to encode a single logical qubit. We further analyze the performance of our extension of the dual-rail encoding scheme with respect to generating a Toffoli gate in Section \ref{subsec:Toffoli}. In Section \ref{sec:lin's H and Z}, we present our second scheme, which show how to implement the Hadamard gate $(H)$ and the phase-flip $(Z)$ simultaneously using only permutations. 
This demonstrates for the first time how a permutational Hadamard could be implemented with other Pauli permutational gates.
In Section \ref{subsec:all24gates}, we investigate the possibility of permutational single-qubit Clifford gates, specifically on the minimum number $M$ of physical qubits to generate the full set of single-qubit Clifford gates by permuting the underlying qubits. In this section we present a non-trivial bound of $M$, ruling out schemes that
can perform permutational Cliffords on too few qubits. 
We also give a set of necessary and sufficient conditions to obtain $M$, which paves the way forward in finding permutational gates that implement the Clifford gates. Finally we summarize and discuss our work in Section \ref{sec:dis}.

\section{Results}
\label{sec:re}
To explain how our scheme works in a DFS, we first define the system's Hamiltonian. We are interested in a system with $N$ modes, each mode of which can be described by identical quantum harmonic oscillators. Explicitly, we can write the Hamiltonian as $H = \sum_{i=1}^N a_i ^\dagger a_i$,
where $a_i$ denotes the lowering operator for the $i$th mode.
Such a Hamiltonian for example is compatible with a physical system comprising of photons of identical frequencies within a quantum bus,
and can also be engineered from any coupled quantum harmonic oscillator system via application of dynamical decoupling pulse sequences \cite{heinze2018universal}.
Note that we can denote $|x_1\> \otimes \dots \otimes |x_N\>$ as a quantum state with $x_i$ excitations in the $i$th mode. Then we call $x_1+ \dots + x_N$ denote the total excitation number of such a state. 
If the number of excitations per mode is at most one, then $(x_1,\dots, x_N)$ are just binary vectors, and
the corresponding excitation number is just their Hamming weight.
It is easy to see that the eigenspaces of $H$ are spanned by states $|x_1\> \otimes  \dots \otimes |x_N\>$ with a constant total excitation number. 
Schemes based on such constant-excitation subspaces have also recently been studied \cite{BvL16,ouyang2018permutation}.
 
Up to a global phase, single-qubit Clifford gates are given by 24 matrices $A_{ij}$. When $j=1,2,3,4$ we have
$   A_{1j}=\qbmat{1}{0}{0}{i^j}$ and    $A_{2j}=\qbmat{0}{1}{i^j}{0}.$ 
Moreover   
$A_{5j}=\frac 1 {\sqrt 2} (A_{1j}+iA_{2j})$ 
and
$A_{6j}=\frac 1 {\sqrt 2} (A_{1j}-iA_{2j})$.
Furthermore
\begin{align}
   A_{31}&=\frac 1 {\sqrt 2} (A_{12}+A_{24}),  \quad
   A_{32}=\frac 1 {\sqrt 2} (A_{14}+A_{22}),  \\
   A_{33}&=\frac 1 {\sqrt 2} (A_{11}+A_{23}),  \quad
   A_{34}=\frac 1 {\sqrt 2} (A_{13}+A_{21}),  
  \end{align}
  and
  \begin{align}
   A_{41}&=\frac 1 {\sqrt 2} (A_{12}-A_{24}),  \quad
   A_{42}=\frac 1 {\sqrt 2} (A_{14}-A_{22}),  \\
   A_{43}&=\frac 1 {\sqrt 2} (A_{11}-A_{23}),  \quad
   A_{44}=\frac 1 {\sqrt 2} (A_{13}-A_{21}).
\end{align}
Note that these 24 gates can be generated by the Hadamard gate $H=A_{31}$ and the phase gate $P = A_{11}$.  
In what follows, we will investigate the extent to which these 24 gates can be implemented by permutations.

\subsection{An extended dual-rail encoding scheme}
\label{subsec:exensch}
In the dual-rail encoding scheme previously considered, 
the states $|0\>|1\>$ and $|1\>|0\>$ encode the logical states of a qubit \cite{KLM01}.  
In this scheme, the only non-trivial permutation possible is one 
that swaps the first qubit with the second.
We denote this permutation by $(1,2)_q$, where the subscript $q$ indicates that $(1,2)_q$ applies to qubits. 
The permutation, $(1,2)_q$ is equivalent to the bit-flip operation on the space, because $(1,2)_q |0\>|1\> =  |1\>|0\>$ and $(1,2)_q|1\>|0\> = |0\>|1\>$.
Since there is no other non-trivial permutation available on two qubits, 
only the bit-flip gate can be performed using permutations in the dual-rail encoding scheme.

By extending the dual-rail encoding scheme, it becomes possible to implement more gates by permuting qubits. 
Namely, we map each physical qubit in the dual-rail encoding scheme to a $2n$-qubit state with two rows of $n$ qubits each.
By extending the dual-rail encoding scheme, we can not only implement the logical bit flip, but also logical gates that perform logical $\begin{pmatrix}
1 & 0 \\
0 & e^{2 \pi i / n} \\
\end{pmatrix}$ gates in the logical basis.

Before we proceed to describe our scheme, we wish to emphasize we do not recommend that permutational gates be actually implemented on the physical system. 
This is because permutational gates are challenging to implement in practice, and may take too long to perform, thereby exposing the system to decoherence.
In our scheme, the permutational gates are only meant to be carried out by a classical computer. 
The classical computer computes the permutations of the labels on the underlying physical qubits. 
The non-permutational gates, which can be carried out by resonant couplings, 
are then performed between appropriate pairs or subsets of relabeled qubits. 
This can be achieved in physical systems where long-range interactions may be possible, such as in ion-trapped quantum architectures or carefully designed superconducting quantum circuits.

The basis states in our extended dual-rail encoding scheme will comprise of $4n$ qubits, with four rows with $n$ qubits each. 
The basis states can be defined in terms of the states $|(j)_n\> = X_j|0\>^{\otimes n}$,
where $X_j$ denotes the $n$-qubit matrix that applies the bit-flip operator on the $j$th qubit and leaves the remaining qubits unchanged.
In particular, instead of using two physical qubits in the dual-rail encoding scheme, we use $4n$ qubits to encode a logical qubit. 
Denoting $x=e^{2\pi i / n}$ to be a root of unity, 
we denote the states 
\begin{align}
 |\psi_0\> &= \frac 1{\sqrt n} \sum_{j=1}^n x^{j-1}  
 |(j)_{n}\> \otimes |0\>^{\otimes n},
\end{align}
and
\begin{align}
 |\psi_1\> &= 
 |0\>^{\otimes n} \otimes
 \frac 1{\sqrt n} \sum_{j=1}^n x^{j-1}  |(j)_{n}\>.
\end{align}
We interpret both of these states to have qubits on two rows, where each row has $n$ qubits.
The orthonormal basis states of our scheme are then given by
\begin{align} 
|0_L^{XT}\> &= |\psi_0\>|\psi_1\> ,\quad 
|1_L^{XT}\> = |\psi_1\>|\psi_0\>,\label{eq:xt}
\end{align}
which are states on $4n$ qubits arranged in 4 rows with $n$ qubits each.

To implement the bit-flip operation, 
it suffices to apply the permutation 
$ \alpha = \beta \otimes \beta$, 
where $\beta$ is a permutation that swaps the rows in $|\psi_0\>$ and $|\psi_1\>$. 
This is depicted in Figure \ref{fig:alpha}.
Formally, $\beta = (1,n+1)_q (2,n+2)_q...(n,2n)_q$ is a product of swaps, where $(j,n+j)_q$ denotes a swap between the $j$th qubit with the $(n+j)$th qubit.

To implement a logical gate that induces a phase on $|1_L^{XT}\>$ but not on $|0_L^{XT}\>$, 
it suffices to cycle qubits on the third row of the logical qubit leftwards.
We can achieve this formally with the permutation 
$\gamma=I_{2^{2n}} \otimes (n,n-1,...,1)_q,$
where $I_{k}$ denotes a size $k$ identity matrix,
and the permutation $(n,n-1,...,1)_q$ cycles qubits 1 to $n$, where the $j$th qubit is mapped to the $(j-1)$th qubit for all $j=2,\dots,n$, and the first qubit is mapped to the $n$th qubit.
We depict this cyclic permutation $\gamma$ in Figure \ref{fig:gamma}.
Then 
\begin{align}
\gamma |0^{XT}_L\> =  |0^{XT}_L\>
\end{align}
and 
\begin{align}
 \gamma |1^{XT}_L\> =  x|1^{XT}_L\> .    
\end{align}
To understand this more explicitly, note that every qubit on  third row of qubits for the logical $|0^{XT}_L\>$ are all identically equal to $|0\>$, and are thereby left unchanged by any permutation on them.
On the other hand, the third row of qubits in the logical $|1^{XT}_L\>$ have the joint state 
$ \frac 1{\sqrt n} \sum_{j=1}^n x^{j-1}  
 |(j)_{n}\> $,
 which picks up a phase of $x$ when cycled leftwards.
 
 Thus, on the logical basis, $\gamma$ 
 implements the
 $\begin{pmatrix}
 1 & 0 \\
 0 & x \\
 \end{pmatrix}$,
 which is equivalent to the $T$ gate
 when $n=8$, because in this scenario, $x =  e^{2\pi i /8} = e^{\pi i /4}$.

\begin{figure}[htbp]
\centering
\subfigure[The permutation $\alpha$ that implements a bit-flip gate.]{
\begin{minipage}[t]{0.5\linewidth}
\centering
\includegraphics[width=1.6in]{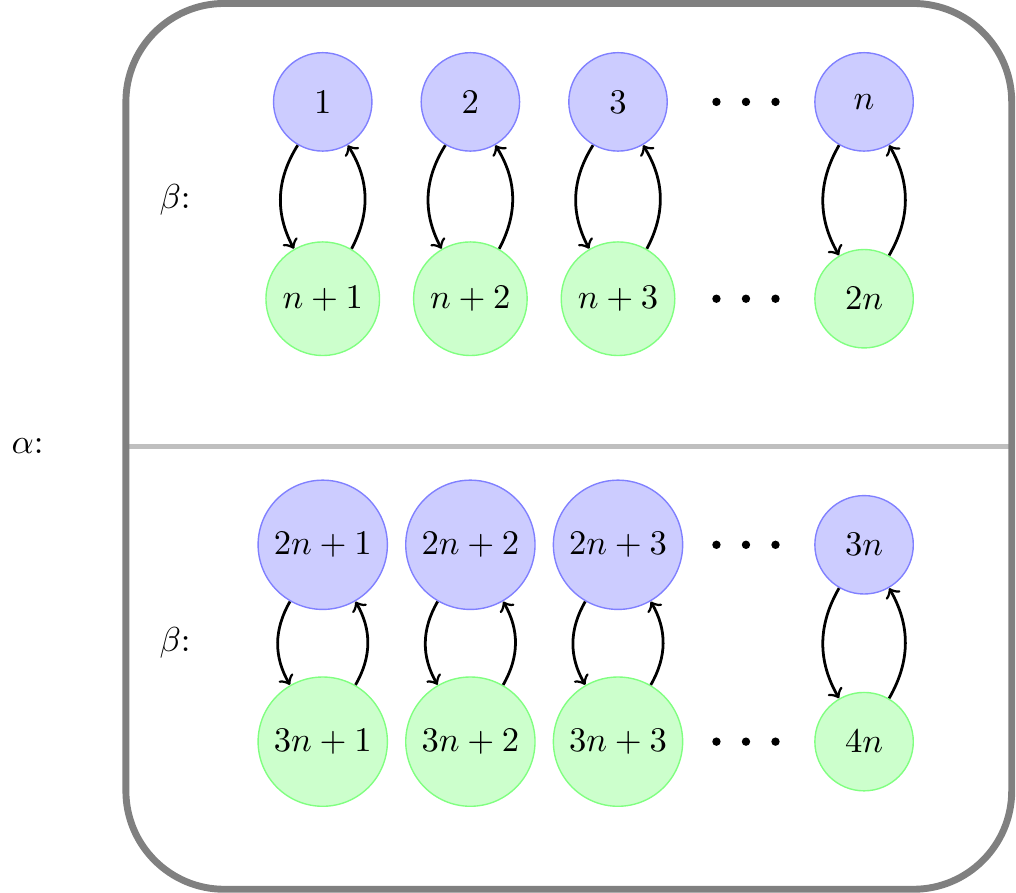}
\label{fig:alpha}
\end{minipage}%
}%
\subfigure[The Hadamard uses $U_{\beta}$, where resonances occur in parallel.]{
\begin{minipage}[t]{0.5\linewidth}
\centering
\includegraphics[width=1.6in]{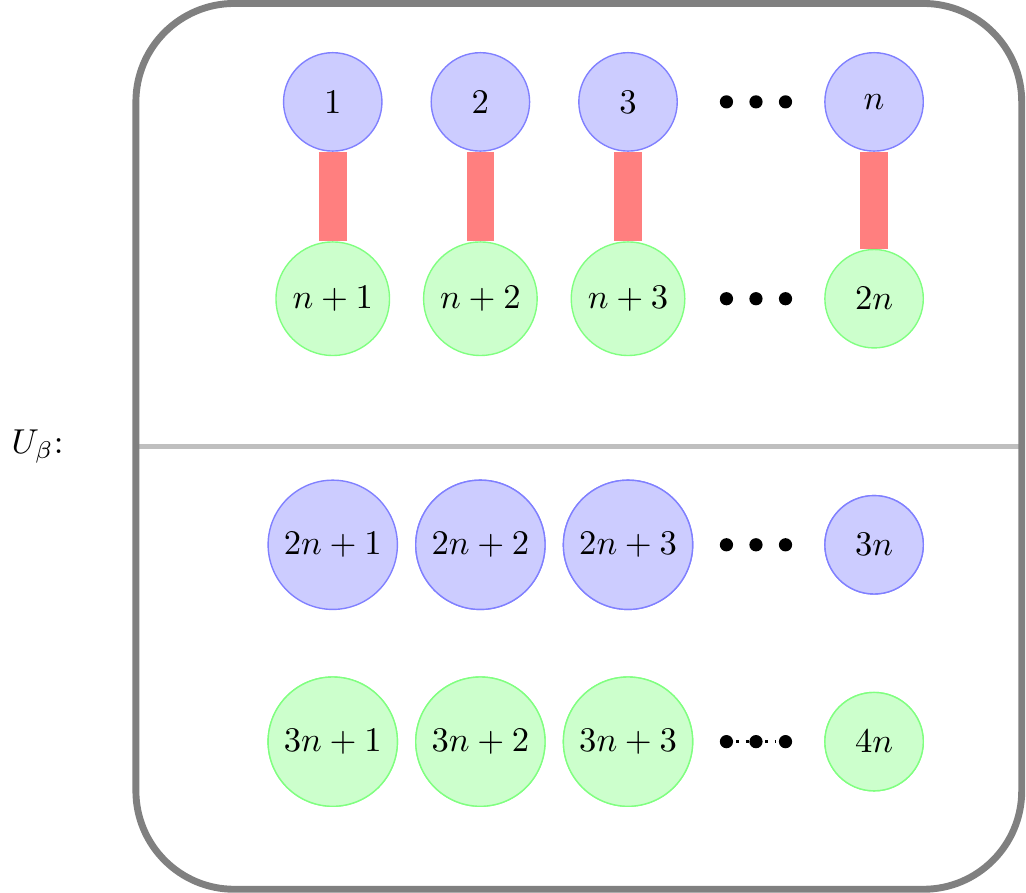}
\label{fig:ubeta}
\end{minipage}%
}%

\subfigure[The permutation $\gamma$ that implements a 
$ |0^{XT}_L\> \<0^{XT}_L| +   e^{2 \pi i /n} |1^{XT}_L\> \<1^{XT}_L|$ gate.]{
\begin{minipage}[t]{0.8\linewidth}
\centering
\includegraphics[width=2.5in]{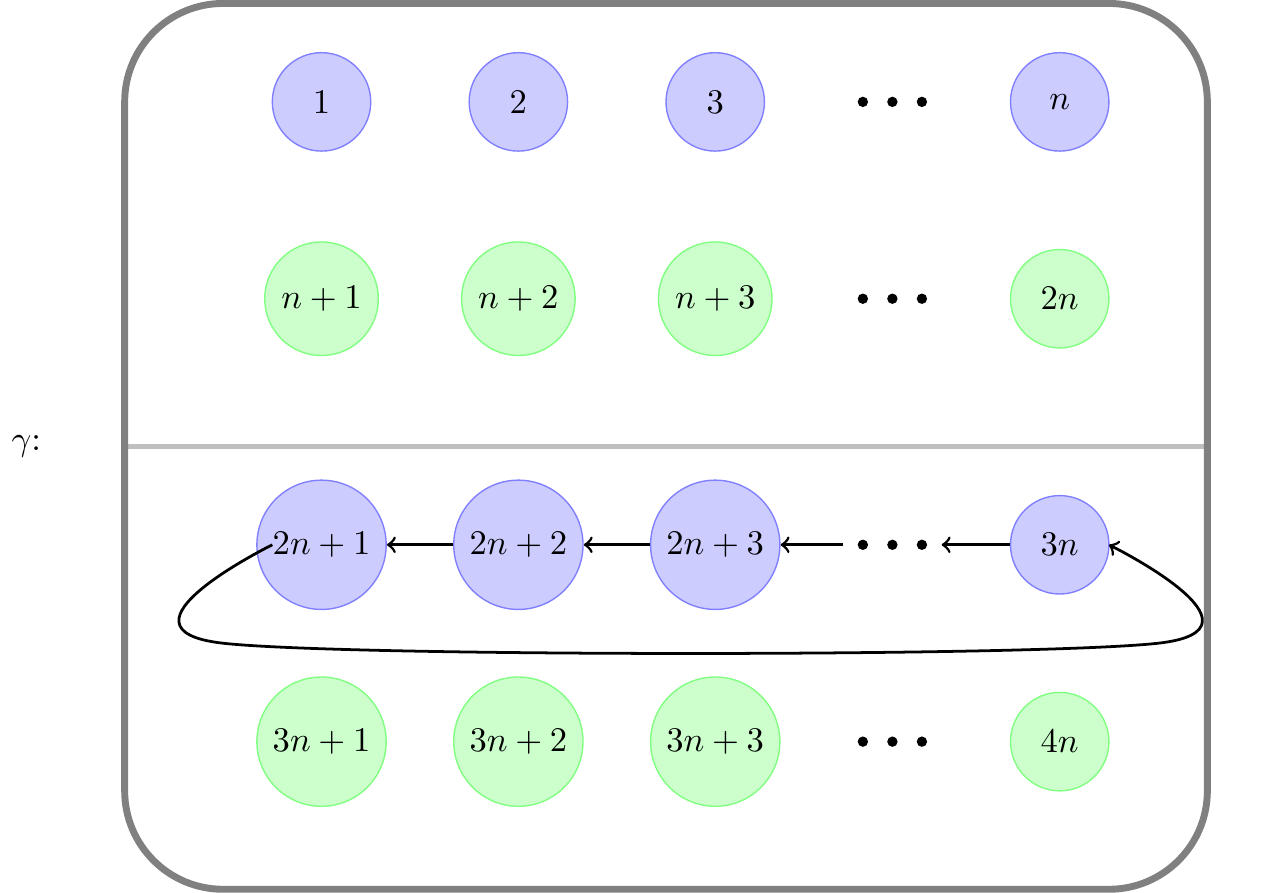}
\label{fig:gamma}
\end{minipage}
}%
\centering
\caption{Some quantum gates constructed from permutations.}
\end{figure}

 Therefore, with 32 qubits and with the logical states in Eq.~(\ref{eq:xt}), we can use the permutations $\alpha$ and $\gamma$ to implement the bit-flip gate and the $T$ gate respectively. When $n=4$ and we have 16 qubits, we can implement a subgroup of the Clifford group as shown in Table \ref{tab:PX}. We will see that the larger $n$ is, the more gates we can implement by permutations.

\begin{table}[h]
\centering
\caption{The eight single-qubit Clifford gates $A_{ij}$ generated by $P$ gate and $X$ gate when $n$ is a multiple of 4.}\label{tab:PX}
\begin{tabular}{|l|cccc|}
\hline
\diagbox{i}{$A_{ij}$}{j} & 1 & 2 & 3 & 4\\
\hline
1 & $P$ & $P^2$ & $P^3$ & $P^4$\\
2 & $PX$ & $(PX)^2$ & $(PX)^3$ & $(PX)^4$ \\ 
\hline
\end{tabular}
\end{table}

In order to allow for arbitrary single-qubit operations, we also need to implement the Hadamard gate on the basis states given by Eq.~(\ref{eq:xt}).
Here, we can no longer rely on permuting the underlying physical qubits. 
Instead, we must utilize resonant couplings which induce Rabi oscillations between specific two-level systems of our choosing.
Specifically, the main idea of how to implement the logical Hadamard arises from a simple observation with respect to the dual-rail encoding scheme. 
For the dual-rail encoding scheme, it suffices to implement the Hamiltonian $X_1 X_2$ and appropriate phase gates. The Hamiltonian $X_1X_2$ in turn is equivalent to implementing the Hamiltonian $X_1$ sandwiched between two CNOT gates.

The permutation $\beta$ which acts as a bit-flip gate between the state $|\psi_0\>$ and $|\psi_1\>$ is essentially a swap of two rows of qubits. 
It is thus a many-body Hamiltonian, and is challenging to implement in practice. 
Fortunately, as we will show, a Hamiltonian $\beta$ is effectively simulated by pairwise resonant couplings between 
pairs of qubits on a subspace our states reside in.
The first type of resonant coupling we use must induce Rabi oscillations between the two-qubit states $|0\>|1\>$ and $|1\>|0\>$. This can be achieved using the Heisenberg exchange interaction $(1,2)_q$
An important technical observation that we use is that a many-body Pauli-type interaction with exchange interactions on a single-excitation subspace 
$B_k = \{ X_1 |0\>^{\otimes k}, \dots, X_k |0\>^{\otimes k} \}   $
can be parallelized as given in Theorem \ref{thm:parallel-XX-16}.
\begin{theorem} \label{thm:parallel-XX-16}
Let $|\psi\>$ be any state in the span of $B_{2n}$, and let $\theta$ be any real number.
Then 
\begin{align}
e^{i\theta \pi_{12}}
e^{i\theta \pi_{34}}
\dots
e^{i\theta \pi_{(2n-1)(2n)}} 
|\psi\>
=
e^{(n-1)i\theta}
e^{i\theta \pi_{12}\pi_{34} \dots \pi_{(2n-1)(2n)}} |\psi\>, \notag
\end{align}
where $\pi_{ij}=(i,j)_q$ is an operator that swaps qubit $i$ and qubit $j$.
\end{theorem}
We prove Theorem \ref{thm:parallel-XX-16} in Appendix \ref{sec:proofth1}.
Because of this theorem, we can decompose $U_\beta$ as illustrated in Figure \ref{fig:ubeta} into a product 
\begin{align}
U_\beta = e^{-(2^n-1)i  \frac{ \pi}{4}} e^{i (1,n+1)_q  \frac{ \pi}{4}} \dots e^{i (n,2n)_q  \frac{ \pi}{4}} \otimes I_{2^{2n}}.
\end{align}
It is therefore possible to implement $U_\beta$ in a single timestep by parallel use of Heisenberg exchange couplings, which makes it possible to implement in realistic physical settings.
We will soon see that the Hadamard gate relies on sandwiching the unitary $U_\beta=e^{i \beta  \frac{ \pi}{4}}\otimes I_{2^{2n}}$
between CNOT-like gates on the states $|\psi_i\>|\psi_j\>$, which we denote as $C_\beta$.
In particular we would like 
\begin{align}
    C_{\beta}
    |\psi_i\>\otimes |\psi_j\> 
    &= 
    |\psi_i\>\otimes |\psi_{(i+j\mod 2)}\>,\quad i,j=0,1.
\end{align}

We now propose how $C_{\beta}$ can be implemented using 
quantum Fredkin gates, {\em i.e.} controlled swap operations.
An advantage of relying on quantum Fredkin gates is that they can be implemented using an effective Hamiltonian 
$g |1\>\<1| (1,2)_q$. Similar interactions can be implemented using superconducting qubits for example \cite{liu2018onestepFredkinSC}. 
This effective Hamiltonian allows Rabi oscillations to be induced between a pair of states $|0\>|1\>$ and $|1\>|0\>$ conditioned on the state of the control qubit, while leaving $|0\>|0\>$ and $|1\>|1\>$ unchanged.
Now denote $F_{a,(b,c)}$ as a quantum Fredkin gate with the $a$th qubit as the controlling qubit, and the swap occurring between the $(2n+b)$th and $(2n+c)$th qubits respectively.  
Specifically, we propose that 
\begin{align}
    C_\beta = \prod_{j=1}^n  C_{j,\beta}
\end{align}
where
\begin{align}
    C_{j,\beta} = F_{n+j,(1,n+1)} \dots F_{n+j,(n,2n)}  
\end{align}
is an operator that swaps the third and fourth rows within the logical qubit conditioned on the value of the $j$th qubit in the second row.
Note that $C_{j,\beta}$ takes $n$ timesteps to proceed, 
and $C_\beta$ can be arranged such that $n^2$ gates occur in $n$ timesteps.

Then it follows that $\gamma^{-2}  C_{\beta}  U_{\beta} C_{\beta}$ acts as a Hadamard on the basis states given by Eq.~(\ref{eq:xt}) because for $j=0,1$ we have 
\begin{align}
    \gamma ^{-2} C_{\beta}  U_{\beta} C_{\beta} |j^{XT}_L\> 
    &=
    \frac{ |0^{XT}_L\> + (-1)^j   |1^{XT}_L\> }{\sqrt 2}.
\end{align}
Hence to implement the Hadamard gate, we require 
$2n+1$ time steps, where $n$ timesteps arise from $C_\beta$, one time step arises from $U_\beta$, and another $n$ arise from $C_\beta$.

We now explain why the $U_\beta$ and $C_\beta$ gates can be implemented within a decoherence-free subspace of the physical Hamiltonian. 
Explicitly, the physical Hamiltonian that we consider is a sum of identical artificial atoms (or quantum harmonic oscillators). 
What this means is that this Hamiltonian commutes with any swap operator on qubits.
Since $U_\beta$ relies on Hamiltonians that are swaps on distinct pairs of qubits, $U_\beta$ must proceed with the DFS of the physical system.
For a similar reason, the operator 
$|1\>\<1| ( 1,2)_q$ commutes with the physical Hamiltonian on the subspace where the quantum computation takes place.
This is because the control part $|1\>\<1|$ is just diagonal in the basis of the physical Hamiltonian, and 
the physical Hamiltonian being permutation-invariant commutes with any swap operation.

To make universal quantum computation possible with our extended dual-rail scheme, we need to show how to perform an entangling gate. 
It turns out that we can perform a CNOT between two logical qubits defined on the basis states (\ref{eq:xt}), by 
appropriate controlled swaps of the underlying rows of qubits.
Now we consider quantum Fredkin gates with controls within the first logical qubit and swaps occurring on the second logical qubit. In particular, let $C_{i,j,k}$ denote a quantum Fredkin gate that swaps the $j$th and $k$th qubits on the target logical qubit. conditioned on the value of the $i$th qubit in the control logical qubit.
Then the logical CNOT can be achieved by doing the following.
\begin{enumerate}
    \item Apply $C_{n+i,j,n+j}$ for $i,j=1,\dots,n$.
    \item Apply $C_{2n+i,j,n+j}$ for $i,j=1,\dots,n$.
\end{enumerate}
This procedure swaps the first and second rows, and the third and fourth rows of the target logical qubit.
These  swaps are conditioned to swapped conditioned on whether the second and third rows in the control logical qubit are all in the all $|0\>$ state.
Hence our logical CNOT uses $2n^2$ quantum Fredkin gates in total, and can be achieved in $n$ timesteps by running appropriate quantum Fredkin gates in parallel.

\subsection{Timesteps required for the Toffoli gate}
\label{subsec:Toffoli}
Here, we analyze the performance of our extension of the dual-rail encoding scheme with respect to generating a Toffoli gate. 
This Toffoli gate can be decomposed into CNOTs and T gates, and it is illustrated in Figure \ref{fig:Toffoli}.

\begin{figure}[ht]
\centering
\includegraphics[width=3.2in]{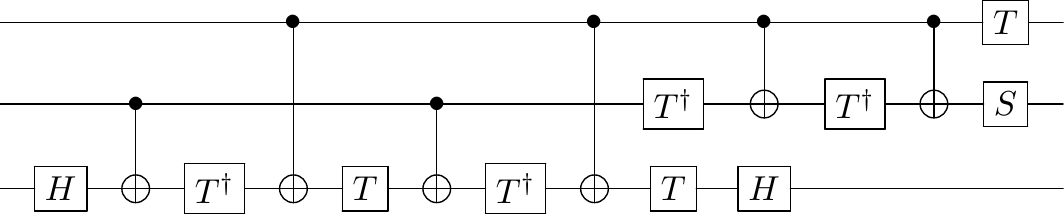}
\caption{The Toffoli gate constructed from CNOTs and T gates}
\label{fig:Toffoli}
\end{figure} 
Two $H$ gates and six CNOTs are required in this setup.
The total number of $S, T$ and $T^\dagger$ gates, which are permutational in our scheme, is $8$.
Hence the number of timesteps required in our extended dual-rail encoding scheme is  $6(n)+2(2n+1) = 10n+2 = 82$ for $n=8$. 
We again emphasize that the permutations corresponding to the $T$, $T^\dagger$ and $S$ gates in our scheme are kept track of by a classical computer, and subsequently, 
and the resonant couplings are carried out between appropriately permuted pairs of qubits.

In the scheme of DiVincenzo {\em et al.} \cite{divincenzo2000universal} which operates within a DFS, each single qubit gate requires at least one time step, and each CNOT gate requires 13 timesteps. Hence this scheme requires at least $7 + 6(13) = 85$ timesteps are required for a single Toffoli gate, which takes longer than that of our scheme.

\subsection{Hadamard gates by permutations}
\label{sec:lin's H and Z}

It is considerably more complicated to implement a Hadamard by permutations. Extending the results of the previous section would surely fail, because $H$ and $T$ generate an set of infinite size, while the number of permutations on any fixed number of qubits is always finite. Hence, any scheme which implements a Hadamard gate by permutations necessarily has to be quite different from the extended dual-rail encoding scheme.

Here, we show the gates generated by the Hadamard ($H$) and the phase-flip ($Z$) can be implemented using only permutations. From Table \ref{tab:hz} we can implement 8 such Clifford gates $A_{ij}$ by using $H$ and $Z$ gates simultaneously.  
\begin{table}[h]
\centering
\caption{The eight single-qubit Clifford gates $A_{ij}$ generated by $H$ gate and $Z$ gate}\label{tab:hz}
\begin{tabular}{|l|cccc|}
\hline
\diagbox{i}{$A_{ij}$}{j} & 1 & 2 & 3 & 4\\
\hline
1 &  & $Z$ &  & $Z^2$\\
2 &  & $ZHZH$ &  & $Z^2HZH$\\
3 & $H$ & $H^2ZH$ &  & \\
4 & $H^2ZH^2$ & $HZ$ &  & \\ 
\hline
\end{tabular}
\end{table}
We first consider the construction of permutational Hadamard gates that act on the basis states
\begin{eqnarray}
\label{eq:0h}
|0^H\>=&& 
\frac 1 {\sqrt 2}
(|x_0,y_0\>+|x_1,y_1\>),	
\\
\label{eq:1h}
|1^H\>=&& 
\frac i {\sqrt2} 
(|x_0,y_0\>-|x_1,y_1\>).
\end{eqnarray}
Here, we require the vectors $|x_0\>,|x_1\>,|y_0\>,|y_1\>$ to have unit norm, and $\<x_0|x_1\> = \<y_0|y_1\> = 0$. 
A permutation $\bH$ acts as a Hadamard gate if 
$\bH|0^H\>={|0^H\>+|1^H\>\over\sqrt{2}}$ and 
$\bH|1^H\>={|0^H\>-|1^H\>\over\sqrt{2}}$.
Equivalently, we require  
$\bH |x_0,y_0\>
=  w^{-1}|x_1,y_1\>,$
where $w=e^{{\pi i\over4}}$.
One can verify that these equations holds whenever $\bH=U \otimes Q$, where $U,Q$ are permutations that are Hermitian and also satisfy the equations 
$U|x_0\>= |x_1\>$ and $Q|y_0\>=w^{-1}|y_1\>$.

In our construction, the logical basis vectors are 
spanned by $|x_0\>|y_0\>$ and $|x_1\>|y_1\>$ 
where
\begin{align}
|x_k\> &= \frac{w^{5k}}{\sqrt 8} \sum_{j=1}^8 w^{(-1)^k(j-1)} |(j)_8\>,\notag\\
|y_0\> &= \frac{1}{\sqrt 8} \sum_{j=1}^8 (-w)^{ (j-1)} |(j)_8\> , \quad
|y_1\> = w |x_0\>.
\end{align}
Let us consider the permutations
$U=(1,6)_q(2,5)_q(3,4)_q(7,8)_q,$ 
$Q=(2,6)_q(4,8)_q,$ and	
$R=(1,7)_q(2,6)_q(3,5)_q.$
Disregarding the effects of global phases, the permutational Hadamard and phase-flips are
$\bH = U \otimes Q$ and $ \quad \bZ = R \otimes Q$ respectively.
Similarly, the permutational bit-flip operator and $\bX \bZ$ are  
$\bX  = URU \otimes Q =  (1,5)_q(2,4)_q(6,8)_q \otimes Q$ and 
$\bY  = URUR \otimes I_{2^8} = (1,7,5,3)_q(2,8,6,4)_q \otimes I_{2^8}$ respectively.

\subsection{On the possibility of permutational single-qubit Clifford gates}
\label{subsec:all24gates}

So far, we have shown that it is possible to implement a subgroup of the single-qubit Clifford gates by permuting the underlying physical qubits. 
But can we implement the entire set of single-qubit Clifford gates using permutations alone?
Here, we supply bounds on $M$, where $M$ denotes the minimum number of physical qubits for which the full set of single-qubit Clifford gates can be performed just by permuting the underlying qubits.

It is shown in \cite{planat2010clifford} that the set of all single-qubit Clifford gates modulo the global phase is isomorphic to $S_4$, which is a symmetric group of size 4. 
Now we argue that $M \ge 12$.
Let $M_k$ denote the subspace spanned by $|{\bf x}\>$ for ${\bf x}$ of Hamming weight $k$. Without loss of generality, the subspace $\mathcal C$ spanned by our logical qubit lies within $M_k$. 
Because $\bH$ and $\bP$ must be permutations on qubits, they induce orbits on appropriate subspaces of $M_k$. 
Because $\bH \bP$ and $\bP$ are isomorphic to a 3-cycle and a 4-cycle respectively on $S_4$, there must be bases $W_q = \{|q_i\>\}$ and $W_p \{|p_i\>\}$ within $M_k$ 
of cardinality $3n_1$ and $4n_2$ respectively where (1) $n_1$ and $n_2$ are positive integers, (2) $q_i$ and $p_i$ are binary vectors of weights $k$, and (3) $\bHP |q_i\> = |q_{i \mod 3 n_1}\>$ and $\bP |p_i\> = |p_{i \mod 4 n_2}\>$. 
Let $\mathcal C_{HP}$ denote the span of $W_q$ and $\mathcal C_P$ denote the span of $W_p$.
So $\mathcal C$ must lie within the intersection of $\mathcal C_{HP}$ and $\mathcal C_P$. But if the intersection of $\mathcal C_{HP}$ and $\mathcal C_P$ not equals $\mathcal C_P$, then $\bP$ does not stabilize $\mathcal C$ which is a contradiction. 
Hence we must have $\mathcal C = \mathcal C_{HP} = \mathcal C_P$ and hence $M$ must be a multiple of $12\ {\rm lcm}(n_1 ,n_2).$

We proceed to give a set of necessary and sufficient conditions for implementing the full set of single-qubit Clifford gates on $M$ qubits. 
\begin{theorem}
\label{th:nesuc}
If the full set of single-qubit Clifford gates can be implemented by $M$ qubits, then there exist two permutation matrices $\bf P$ and $\bf H$ with size $2^M$ 
and complex numbers $z_1$ and $z_2$ of modulus one such that
$\mathop{\rm rank}
\begin{pmatrix}
A\\
A^{\prime}
\end{pmatrix}
< 2^{M+1},$ where
\begin{equation}
A=
    \begin{pmatrix}
    \bP-z_1 I &  0 \\ 
    0 & -i \bP-z_1 I \\
    \end{pmatrix},\\
 A^{\prime}=
    \begin{pmatrix}
    \sqrt2 \bH-z_2 I &  -z_2 I \\ 
    -z_2 I & \sqrt2 \bH+z_2 I \\
    \end{pmatrix}. 
\end{equation}
\end{theorem} 
The above can be easily shown because 
the gates $\bP$ and $\bH$ which generate the set of single-qubit Clifford gates
must satisfy the equations 
$\bP {\bf u} = z_1 {\bf u}$,
$\bP {\bf v} = i z_1 {\bf v}$,
$\bH {\bf u} = z_2 {{\bf u}+{\bf v} \over \sqrt 2} $,
$\bH {\bf v} = z_2 {{\bf u}-{\bf v} \over \sqrt 2}$.
Then we can write the above as a matrix equation as
\begin{align}
    \begin{pmatrix}
    \bP &  0 \\ 
    0 & \bP \\
    \end{pmatrix}
    \begin{pmatrix}
     {\bf u} \\
     {\bf v} \\
    \end{pmatrix}
=
    z_1\begin{pmatrix}
     {\bf u} \\
    i {\bf v} \\
    \end{pmatrix},
    \quad
    \begin{pmatrix}
    \bH &  0 \\ 
    0 & \bH \\
    \end{pmatrix}
    \begin{pmatrix}
     {\bf u} \\
     {\bf v} \\
    \end{pmatrix}
=
\frac{z_2}{\sqrt 2}
    \begin{pmatrix}
     {\bf u}  + {\bf v}\\
     {\bf u} - {\bf v}\\
    \end{pmatrix}.
\end{align}
This is equivalent to 
\begin{align}
&
\left(
    \begin{pmatrix}
    \bP &  0 \\ 
    0 & -i \bP \\
    \end{pmatrix}
    -z_1
    \begin{pmatrix}
    I &  0 \\ 
    0 &  I \\
    \end{pmatrix}
    \right)
    \begin{pmatrix}
     {\bf u} \\
     {\bf v} \\
    \end{pmatrix}
=
    0,  \notag  \\
    &\left(
    \begin{pmatrix}
    \bH &  0 \\ 
    0 &  \bH \\
    \end{pmatrix}
    -\frac{z_2}{\sqrt{2}}
    \begin{pmatrix}
    I &  I \\ 
    I &  -I \\
    \end{pmatrix}
    \right)
    \begin{pmatrix}
     {\bf u} \\
     {\bf v} \\
    \end{pmatrix}
=0.
\end{align}
It follows that any non-zero solution for ${\bf u}$ and ${\bf v}$ yields a logical basis for our single-qubit Clifford gates. Since we only focus on the non-zero solutions, $A$ and $A'$ are not unique. For example, one can also suppose $A=
    \begin{pmatrix}
    \bP-z_1 I &  0 \\ 
    0 & \bP-iz_1 I \\
    \end{pmatrix}$.
Hence it suffices to find a non-zero intersection of kernels of $A$ and $A'$. 
This non-zero intersection in turn occurs if and only if 
\begin{eqnarray}
\mathop{\rm rank} \begin{bmatrix} A\\ A'\end{bmatrix} < 2^{M+1}.
\end{eqnarray}
which is a standard fact in matrix analysis \cite[Fact 2.11.3]{Bernstein:2009:MMT}.  
We however leave the problem of obtaining an upper bound on $M$ open.

\section{Discussion}
\label{sec:dis}
Expediting quantum computation on schemes that are inherently protected against noise is a tantalizing prospect.
Such schemes have been explored by DiVincenzo {\em et al.} \cite{divincenzo2000universal} on spin-based quantum computers, and more recently, also on topological quantum computers \cite{freedman2003topological}.
Utilizing permutations is one approach that can potentially speed up quantum computations that has been explored recently \cite{Guanyu1711.05752,Guanyu1806.02358,Guanyu1806.06078,planat2010clifford,planat2017magic,Planat1808.06831}.
However these quantum computations either operate on non-abelian anyons \cite{Guanyu1711.05752,Guanyu1806.02358,Guanyu1806.06078}, focus on just the group structure of permutational subgroups related to quantum computations \cite{planat2010clifford,Planat1808.06831}, or work on the magic state model of quantum computation related to permutations \cite{planat2017magic}.
In this sense, prior work has neither addressed the effect of permutations on the underlying qubits nor addressed the case when the qubits are regular fermions or bosons which are more abundant in an experimental setting.

In this report, we fill this gap by showing explicitly how to perform certain single-qubit quantum computations by simply permuting the underlying qubits. Together with exchange-interactions and other resonant couplings, our scheme allows a faster implementation of the Toffoli gate. 
Our scheme can be seen to be a permutational extension of the dual-rail encoding scheme \cite{KLM01}, and by virtue of being supported on encoded qubits with a low excitation number, 
allows universal gate set by using simple resonant interactions.
We also explore the possibility of implementing other single qubit gates by permutations, 
and give necessary and sufficient conditions for their realization.

We believe that determining if all the single-qubit Clifford gates can be realized on a DFS with only a single-excitation is an important problem. 
This is because if this were possible, exchange type couplings can implement a CNOT gate in parallel, and make possible an arbitrary Clifford computation on any number of logical qubits without requiring use of the Fredkin gate. 
Given that Clifford computations are known to be hard under reasonable computation assumptions \cite{yoganathan2018quantum},
this would give rise to a way to realize speedy Clifford computations using a simple scheme,  
and could bring us closer to the demonstration of quantum supremacy.

\section*{Acknowledgements}
We thank the referees for their feedback and recommendations.
YO acknowledges support from the Singapore National Research Foundation under NRF Award NRF-NRFF2013-01, the U.S. Air Force Office of Scientific Research under AOARD grant FA2386-18-1-4003,
and the Singapore Ministry of Education.
YO is supported by the EPSRC (grant no. EP/M024261/1), and also acknowledges the support of the QCDA project
which has received funding from the QuantERA ERANET Cofund in Quantum Technologies implemented within the European Union’s Horizon 2020 Programme.
LC and YS were supported by the NNSF of China (Grant No. 11871089), and the Fundamental Research Funds for the Central Universities (Grant Nos. KG12080401 and ZG216S1902).

\bibliography{HnP}{}
\bibliographystyle{ieeetr}

\appendix
\begin{widetext}

\section{A parallel construction of an effective many-body Hamiltonian}
\label{sec:proofth1}
Here, we will prove Theorem 1 in the main text. First we present a supporting lemma as follows. 

\begin{lemma} \label{lem:composition}
Let $\theta$ be any real number, and $P_{ij}$ be the permutation matrix which swaps the $i$th and $j$th rows. 
Then 
$ 
e^{i \theta P_{12} }
e^{i \theta P_{34} } 
= 
e^{i \theta}
e^{i \theta P_{12} P_{34} }$.
Moreover for every integer $n$ for which $n \ge 3$ we also have 
$e^{i \theta P_{12}\dots P_{(2n-3)(2n-2)} }
e^{i \theta P_{(2n-1)(2n)} } 
= 
e^{i \theta}
e^{i \theta P_{12} \dots P_{(2n-1)(2n)} }$.
\end{lemma}
\begin{proof}
One can verify $\big[P_{12}P_{34}\dots P_{(2n-3)(2n-2)}, P_{(2n-1)(2n)}\big]=0$ for any $n\geq 2$. We first show the following identity.
\begin{align}
    P_{12}P_{34}\dots P_{(2n-3)(2n-2)} + P_{(2n-1)(2n)} &= P_{12}P_{34}\dots P_{(2n-3)(2n-2)}P_{(2n-1)(2n)}+I_{2n},\quad \forall n\geq 2. \label{eq:perm-identity}
\end{align}
By computing we obtain $P_{12}P_{34}\dots P_{(2n-1)(2n)}=\underbrace{\sigma\oplus\dots\oplus\sigma}_n$, where $\sigma=\begin{bmatrix}
0 & 1\\
1 & 0
\end{bmatrix}$.
Therefore, we obtain
\begin{equation}
\label{eq:identity}
\begin{aligned}
P_{12}P_{34}\dots P_{(2n-3)(2n-2)} + P_{(2n-1)(2n)}&=(\underbrace{\sigma\oplus\dots\oplus\sigma}_{n-1}\oplus I_2)+(I_{2n-2}\oplus \sigma)\\
                                                   &=(\underbrace{\sigma\oplus\dots\oplus\sigma}_n) + I_{2n} \\
                                                   &=P_{12}P_{34}\dots P_{(2n-1)(2n)}+ I_{2n}.
\end{aligned}
\end{equation}
Then it follows that 
\begin{equation}
\label{eq:lemma3}
e^{i \theta P_{12} }
e^{i \theta P_{34} } 
= e^{i \theta P_{12} + i \theta P_{34}}
= e^{i \theta P_{12} P_{34} +i \theta I_4}
= e^{i \theta} e^{i \theta P_{12} P_{34} },
\end{equation}
where the first equality uses \eqref{eq:perm-identity} and the second equailty uses \eqref{eq:identity}.
Similarly the result for $n \ge 3$ follows, and we complete the proof of the lemma.
\end{proof}
Using the above lemma then we prove Theorem \ref{thm:parallel-XX-16} as follows. 
\begin{proof}
Using Lemma \ref{lem:composition} and the fact that the swap operators $\pi_{ij}$ on qubits correspond precisely to the permutation of qubit labels, we obtain
\begin{equation}
\label{eq:th1proof}
\begin{aligned}
e^{i\theta \pi_{12}}
e^{i\theta \pi_{34}}
\dots
e^{i\theta \pi_{(2n-1)(2n)}} 
&=
e^{i\theta}e^{i\theta\pi_{12}\pi_{34}}
e^{i\theta \pi_{56}} 
\dots
e^{i\theta \pi_{(2n-1)(2n)}} \\
&=
e^{2i\theta}
e^{i\theta \pi_{12}\pi_{34}\pi_{56}}
e^{i\theta \pi_{78}} 
\dots
e^{i\theta \pi_{(2n-1)(2n)}} \\
&=\cdots  \\
&=e^{(n-1)i\theta}e^{i\theta \pi_{12}\pi_{34}\dots\pi_{(2n-1)(2n)}}.
\end{aligned}
\end{equation}
This completes the proof.
\end{proof}

\end{widetext}

\end{document}